\documentclass[a4]{article}
\usepackage{bm}
\usepackage[all]{xy}
\usepackage{graphicx}
\usepackage{verbatim}
\usepackage{wrapfig}
\usepackage{ascmac}
\usepackage{makeidx}
\usepackage{amscd}
\usepackage{url}
\usepackage{comment}
\usepackage{enumerate}
\usepackage{here}
\usepackage{latexsym}
\usepackage{array}
\usepackage{booktabs}
\usepackage{multirow}
\usepackage{mathrsfs}
\usepackage{geometry}
\usepackage{fancyhdr}
\renewcommand{\title}[1]{

\begin{center} \Large \bf #1 \end{center}
}

\renewcommand{\author}[2]{
 \begin{center} #1  \vspace{3mm} \\
  #2 \\
 \end{center}
\addvspace{\baselineskip}
}

\usepackage{amssymb}
\usepackage{amsmath}

\usepackage{amsthm}
\newtheorem{theorem}{Theorem}[section]
\newtheorem{proposition}[theorem]{Proposition}
\newtheorem{corollary}[theorem]{Corollary}
\newtheorem{lemma}[theorem]{Lemma}
\newtheorem{example}[theorem]{Example}
\theoremstyle{definition}
\newtheorem{definition}[theorem]{Definition}

\theoremstyle{remark}

\makeatletter
\@addtoreset{equation}{section}

\makeatother

\begin{document}
\baselineskip 5mm
\title{Categorical Perspective on\\ Quantization of Poisson Algebra}
\author{${}^1$ Jumpei Gohara ${}^2$ Yuji Hirota and~ ${}^1$ Akifumi Sako}
{
${}^1$  Tokyo University of Science,\\ 1-3 Kagurazaka, Shinjuku-ku, Tokyo, 162-8601, Japan\\
${}^2$
Azabu University,\\ 1-17-71 Fuchinobe, Chuo-ku, Sagamihara, Kanagawa, 252-5201, Japan}
\noindent
\vspace{1cm}

\abstract{
We propose a generalization of quantization using a categorical approach. For a fixed Poisson algebra, quantization categories are defined as subcategories of the $R$-module category equipped with the structure of classical limits. We then construct the generalized quantization categories including matrix regularization, strict deformation quantization, prequantization, and Poisson enveloping algebra. It is shown that the categories of strict deformation quantization, prequantization, and matrix regularization with certain conditions are equivalent categories. On the other hand, the categories of Poisson enveloping algebra are not equivalent to the other categories. }
%
%
%
\section{Introduction}
Noncommutative geometry is regarded as one of the key concepts for formulating the
quantum gravity theory or non-perturbative string theory.
There are many ways to construct noncommutative geometry, 
including deformation quantization, geometric quantization, $C^*$-algebra,
matrix regularization, and so on. 
To find the best approach for quantum gravity or other physics, a unified perspective and a more general formulation containing the existing quantization models would be useful.
\bigskip

In this article, we define a generalized quantization of a Poisson algebra as a subcategory of the category of modules over a commutative algebra.
It is shown that matrix regularization, strict deformation quantization, and prequantization 
are included in the generalization, and each pair of them are equivalent categories under some conditions described later.
In addition, universal enveloping algebra derived from a Poisson algebra
is also formulated as the generalized quantization method.\\
\bigskip

In preparation for the following sections, we review several definitions of noncommutative geometries or quantizations.\\

Dirac introduces the quantization rule as replacing the Poisson brackets by commutators. 
When we regard quantization as a map $~\hat{}~$ from functions to operators
acting linearly on a Hilbert space, it is generally considered that
the axioms for quantization maps ~$\hat{}$~ satisfy the following conditions:
(1)$ \widehat{(H_1 + H_2)} ={\hat{H}_1} + {\hat{H}_2} $.~
(2)$
{\widehat{(\lambda H)}} = \lambda {\hat H}, \quad \lambda \in {\mathbb R}
$.~
(3)$
[ \hat{H}_1, \hat{H}_2 ] = i \widehat{\{ H_1 , H_2 \}}
$
where
$[\hat{H}_1, \hat{H}_2] = \hat{H}_1 \hat{H}_2 - \hat{H}_2 \hat{H}_1$.~
(4)
${\hat 1}= Id$  ($1$ is a constant $1$ and $Id$ is an identity operator.)
(5)
${\hat q^i }$ is the multiplication operator of the function $q^i$, and 
${\hat p^i }= \frac{1}{i} \frac{\partial}{\partial q^i} $.~
However, because no theory satisfies all these conditions, each quantization is defined under weaker conditions or with some restrictions.\\

Let us consider the definition of matrix regularization of a symplectic manifold $(M,~\omega)$. Matrix regularization \cite{matrix1} has evolved from the ideas of Berezin-Toeplitz quantization \cite{berezin1,berezin2}, Fuzzy space \cite{fuzzy1}, and so on. We employ the definition by \cite{arnlind}.
\begin{definition}\label{matrixreg}
Let $N_1,N_2,\ldots $ be a strictly increasing sequence of positive integers and $\hbar$ be a real-valued strictly positive decreasing function such that $\lim_{N\to \infty}N\hbar(N)$ converges. Let $T_k$ be a linear map from $C^\infty(M)$ to $N_k\times N_k$ Hermitian matrices for $k=1,2,\ldots$. If the following conditions are satisfied, then we call the pair $(T_k,~\hbar)$ a $C^1$-convergent matrix regularization of $(M,~\omega)$.
\begin{enumerate}
  \item $\displaystyle \lim_{k\to \infty}\|T_k(f)\|<\infty$, \label{qcon1}
  \item $\displaystyle \lim_{k\to \infty}\|T_k(fg)-T_{k}(f)T_{k}(g)\|=0$, \label{qcon2}
  \item $\displaystyle \lim_{k\to \infty}\|\frac{1}{i\hbar(N_k)}[T_k(f),T_k(g)]-T_k(\{f,g\})\|=0$, \label{qcon3}
  \item $\displaystyle \lim_{k\to \infty}2\pi\hbar(N_k){\rm Tr}T_k(f)=\int_M f\omega$,
\end{enumerate}
where $\|~\|$ is the operator norm, $\omega $ is a symplectic form on $M$ and $\{~,~\}$ is the Poisson bracket induced by $\omega$.
\end{definition}\par
\bigskip
%
%
Formal deformation quantization is defined as follows \cite{bayen1,bayen2,DeW-Lec,Fedosov,Kontsevich,Omori}.
\begin{definition}
Let $\cal F$ be a set of formal power series in $\hbar$ with
coefficients of $C^{\infty}$
functions on Poisson or Symplectic manifold $M$
\begin{eqnarray}\label{formalpower}
{\cal F} := \left\{  f \ \Big| \ 
f = \sum_k \hbar^k f_k, ~f_k \in C^\infty (M)
\right\} ,
\end{eqnarray}
where $\hbar$ is a noncommutative parameter.
A star product is defined on ${\cal F}$ by 
\begin{eqnarray}
f * g = \sum_k \hbar^k C_k (f,g), 
\end{eqnarray}
such that the product satisfies the following conditions.
\begin{enumerate}
\item $*$ is an associative product.
\item $C_k$ is a bidifferential operator.
\item $C_0$ and $C_1$ are defined as 
\begin{eqnarray}
&& C_0 (f,g) = f g,  \\
&&C_1(f,g)-C_1(g,f) = i \{ f, g \}, \label{weakdeformation}
\end{eqnarray}
where $\{ f, g \}$ is the Poisson bracket.
\item $ f * 1 = 1* f = f$.
\end{enumerate}
\end{definition}
Several variations of the deformation quantization with some minor changes from this definition exist. In this definition, the algebra is treated as a set of formal power series of smooth functions. For an arbitrary Poisson manifold, there exists a deformation quantization \cite{Kontsevich}. The formal deformation quantization has been widely adopted. However, it is difficult to regard the deformation quantization as a theory of physics when the theory remains formal. We thus employ a strict deformation quantization introduced by Rieffel \cite{rieffel1,rieffel2}. There are various definitions of this quantization. We use a definition similar to that of \cite{strict2} in this article. Before defining the strict deformation quantization, therefore, we will consider the definition of strict quantization in \cite{strict2}.
\begin{definition}\label{strict1}
Let $\mathcal{A}_0$ be a Poisson algebra which is densely contained in the self-adjoint part $\mathcal{C}^0_{\mathbb{R}}$ of an abelian $C^*$-algebra $\mathcal{C}^0$. Let $I$ be a subset of real numbers which contains $0$. Strict quantization of the Poisson algebra $\mathcal{A}$ is a family of maps $(\mathcal{Q}^\hbar:\mathcal{A}_0\to \mathcal{C}^\hbar_{\mathbb{R}})$, where
\begin{enumerate}
  \item $\mathcal{C}^\hbar$ is a $C^*$-algebra with an associative product $\times_\hbar$ and a $C^*$-norm $\|~\|_\hbar$. For $a,b\in \mathcal{C}^\hbar_{\mathbb{R}}$ which is the self-adjoint part of $\mathcal{C}^\hbar$,
\begin{align*}
a\star_\hbar b&:=\frac{1}{2}(a\times_\hbar b+b\times_\hbar a),\\
[a,b]_\hbar&:=(a\times_\hbar b-b\times_\hbar a).
\end{align*}
  \item $\forall \hbar\in I$, $\mathcal{Q}^\hbar:\mathcal{A}_0\to \mathcal{C}^\hbar_{\mathbb{R}}$ is $\mathbb{R}$-linear and $\mathcal{Q}^0$ is just the inclusion map such that
\begin{enumerate}
  \item[$(a)$] For $f\in \mathcal{A}_0$, the map $\hbar \to \|\mathcal{Q}^\hbar (f)\|_\hbar$ is continuous.
  \item[$(b)$] For $f,g\in \mathcal{A}_0$, $\|\mathcal{Q}^\hbar (f)\star^\hbar \mathcal{Q}^\hbar(g)-\mathcal{Q}^\hbar(fg)\|_\hbar \to 0$ as $\hbar \to 0$.
  \item[$(c)$] For $f,g\in \mathcal{A}_0$, $\|(i\slash \hbar)[\mathcal{Q}^\hbar (f),\mathcal{Q}^\hbar(g)]_\hbar-\mathcal{Q}^\hbar(\{f,g\})\|_\hbar \to 0$ as $\hbar\to 0$.
  \item[$(d)$] $\forall \hbar\in I$, $\mathcal{Q}^\hbar(\mathcal{A}_0)$ is dense in $\mathcal{C}^\hbar_{\mathbb{R}}$.
\end{enumerate}
\end{enumerate}
\end{definition}
We define a strict deformation quantization based on this strict quantization.
\begin{definition}\label{strict2}
If $\forall \hbar \in I$, $\mathcal{Q}^\hbar (\mathcal{A}_0)$ is a subalgebra of $\mathcal{C}^\hbar_{\mathbb{R}}$ and $\mathcal{Q}^\hbar$ is injective, a strict quantization $\mathcal{Q}^\hbar$ is called a strict deformation quantization.
\end{definition}
Especially, a strict deformation quantization of a Poisson manifold $M$ is defined by a Poisson subalgebra $\mathcal{A}_0(M)$ of $C^\infty(M)$ composed of bounded functions.
In this article, we use not formal but strict deformation quantization $(\mathcal{A}_0(M),\mathcal{Q}^\hbar)$.\par
\bigskip
The prequantization is defined as follows. (See for example \cite{prequantization0,prequantization1}. For the general Poisson manifolds case, Vaisman established the geometric quantization \cite{prequantization2}.)
\begin{definition}\label{def_prequantization}
Prequantization for Poisson manifolds is a procedure to assign to each Poisson manifold 
a hermitian line bundle with a connection whose curvature is represented in terms of each Poisson bivector.  
\end{definition}
One finds that the prequantization of a Poisson manifold $M$ gives rise to a representation of $C^\infty(M)$ 
on the space of smooth sections of the hermitian line bundle $K\overset{\pi}{\to} M$, that is, 
a mapping $\hat{}~:C^\infty(M)\to {\rm End}(\varGamma(K))$ satisfying that, for $H_1$, $H_2$ and $H\in C^\infty(M)$,
\begin{enumerate}
\item $ \widehat{(H_1 + H_2)} ={\hat{H}_1} + {\hat{H}_2}, $
\item${\widehat{(\lambda H)}} = \lambda {\hat H}, \quad \lambda \in {\mathbb R}$
\item$[ \hat{H}_1, \hat{H}_2 ] = i \hbar\widehat{\{ H_1 , H_2 \}}$
where
$[\hat{H}_1, \hat{H}_2] = \hat{H}_1 \hat{H}_2 - \hat{H}_2 \hat{H}_1$.
\item
${\hat 1}= Id$,  ($1$ is a constant $1$ and $Id$ is an identity operator.)
\end{enumerate}
The following theorem is important to understand the prequantization through the concrete construction of the quantization map $\hat{}$~.
\begin{theorem}
Let $(M,\omega)$ be a prequantizable symplectic manifold,
$L$ be a line bundle called prequantum line bundle, and
$S$ be the subset of smooth square integrable sections of $L$ with compact support. 
There exist a quantization map 
$C^\infty (M) \rightarrow  Op (S): f \mapsto \hat{f}$
satisfying the above 1-4. Here, $Op (S)$ is a set of linear operators acting on $S$.
The map is constructed concretely:
\begin{equation}\label{6_1_37}
\hat{f}: s \mapsto 
\frac{\hbar}{i}  (\nabla_{X_f})s + fs , 
\end{equation}
where $X_f$ is a Hamilton vector field of $f$ and $\nabla_{X_f} s$ is a covariant derivative of the section $s$ along $X_f$.
\end{theorem} 
\bigskip

In this article, we put forward a new framework of quantization including the above 
quantization theories by using categorical methods.
In Section \ref{secqc}, the category of quantization of Poisson algebras 
is defined. 
In Sections \ref{secmr}, \ref{secdq} and \ref{secpq}, we introduce categories including matrix regularization, strict deformation quantization, and prequantization.
We show that they are the categories of quantization of the Poisson algebras.
In Section \ref{secenv}, we study the universal enveloping algebra derived from Poisson algebras.
It is found that the categories including matrix regularization, strict deformation quantization and prequantizaton are equivalent categories under some conditions in Section \ref{secce}.
In Section \ref{seccon}, we summarize all results and make some remarks. In addition, we discuss applications of the quantization category to physics.

\section{Quantization Category}\label{secqc}
In this section, we give a generalization of quantization as a subcategory of the category of modules. We call this subcategory ``quantization category" and show that some quantization theories are included in this category in the following sections. A limit in the category theory corresponds to a classical limit by considering a sequence of morphisms as quantization maps. Before introducing the quantization category, we define a pre-$\mathscr{Q}$ category $\mathscr{P}(\mathcal{A})$.

\begin{definition}\label{preq}
Let $R${\rm Mod} be a category of $R$-module for a commutative algebra $R$ over $\mathbb{C}$. For a Poisson algebra $\mathcal{A}$, a subcategory $\mathscr{P}(\mathcal{A})$ of $R${\rm Mod} is defined as follows. 
\begin{enumerate}
  \item[$1$] $\mathcal{A}\in ob(\mathscr{P})$
  \item[$2$] For arbitrary $M_i\in ob(\mathscr{P}),$ at least a morphism $T_i \in \mathscr{P}(\mathcal{A},~M_i)$ exists. We call $T_i$ a quantization map.
  \item[$3$] For arbitrary $M_i\in ob(\mathscr{P})$, $M_i$ is an associative algebra with a Lie bracket $[~,~]_i$, where the Lie bracket is the commutator for the  associative product of algebra, i.e. $[A,B]=A\cdot B-B\cdot A$ for $A,B\in M_i$.
  \item[$4$] The Lie bracket $[~,~]_k$ satisfies 
\begin{align}\label{lie}
[T_k(f),T_k(g)]_k=\sqrt{-1}\hbar (T_k)T_k(\{f,g\})+\tilde{O}(\hbar^{1+\epsilon} (T_k)) \quad 
(\epsilon>0 ),
\end{align}
where $\hbar$ is noncommutative complex parameter $\hbar:\coprod _k\mathscr{P}(\mathcal{A},M_k)\cup \mathscr{P}(\mathcal{A},\mathcal{A})\to \mathbb{C}$, and $\tilde{O}(\hbar^{1+\epsilon} (T_k))$ is defined in Appendix \ref{ap1}.
  \item[$5$] For arbitrary $A,B\in ob(\mathscr{P})$, if $T_{AB}\in \mathscr{P}(A,B)$ is a linear isomorphism, then there exist $T_{AB}^{-1}\in\mathscr{P}(B,A)$ which satisfies  \label{concon5}
\begin{align*}
T_{AB}T_{AB}^{-1}=id_{B},\quad T_{AB}^{-1}T_{AB}=id_A.
\end{align*} 
\end{enumerate}
We call the subcategory $\mathscr{P}(\mathcal{A})$ the pre-$\mathscr{Q}$ category.
\end{definition}
Note that the Lie bracket $[f,g]_{\mathcal{A}}$ is not the Poisson bracket but $fg-gf=0$. For $T=id_{\mathcal{A}}\in \mathscr{P}(\mathcal{A},\mathcal{A})$, $\hbar(id_{\mathcal{A}})=0$ from $(\ref{lie})$. The pre-$\mathscr{Q}$ category is not uniquely determined by $\mathcal{A}$. As shown later, there are several pre-$\mathscr{Q}$ categories that contain the same $\mathcal{A}$. For example, even a category that has only object $\mathcal{A}$ is a pre-$\mathscr{Q}$ category. \par
Below we denote $\mathscr{P}(\mathcal{A})$ by $\mathscr{P}$ for simplicity.
\begin{definition}\label{chidef}
A map $\chi:ob(\mathscr{P})\to \mathbb{R} $ is defined by the absolute maximum of $\hbar$:
\begin{align*}
\chi(M_i):=\max_{T_i \in \mathscr{P}(\mathcal{A},M_i)} |\hbar(T_i)|.
\end{align*}
We call the map $\chi$ a noncommutative character.
\end{definition}
Let us define an index category of $\mathscr{P}$.
\begin{definition}\label{jdef}
Let $\mathscr{P}$ be a pre-$\mathscr{Q}$ category. $J^\bullet(\mathscr{P})$ is an index category as disjoint union $J^\bullet(\mathscr{P}):=\coprod _\alpha J^\alpha(\mathscr{P})$, where $J^\alpha(\mathscr{P})$ is a connected component of $J^\bullet(\mathscr{P})$. $\alpha$ is the index of the connected component of the index category $J^\bullet(\mathscr{P})$. $J^\alpha(\mathscr{P})$ is defined as follows.
\begin{itemize}
  \item[$({\rm i})$] A set of objects $ob(\mathscr{P})\backslash \{\mathcal{A}\}$ is one to one correspondence with $ob(J^\bullet(\mathscr{P}))$ i.e. $\exists M_i\in ob(\mathscr{P})\backslash \{\mathcal{A}\}\Leftrightarrow \exists i\in ob(J^\bullet(\mathscr{P}))$.
  \item[$({\rm ii})$] $\forall T_{ij}^k\in \mathscr{P}(M_i,M_j),~\chi(M_i)\le \chi(M_j) \Rightarrow  \exists f^k_{ij}\in J^\alpha(\mathscr{P}) (i,j)$.
  \item[$({\rm iii})$] $\forall f^k_{ij}\in J^\alpha(\mathscr{P})(i,j),~\exists T_{ij}^k\in \mathscr{P}(M_i,M_j),~\chi(M_i)\le \chi(M_j)$.
\end{itemize}
\end{definition}
Note that, for $M_i\in ob(\mathscr{P})\backslash \{\mathcal{A}\}$, $J^\alpha(\mathscr{P})$ uniquely exists such that $i\in ob(J^\alpha(\mathscr{P}))$. Thus $J^\bullet(\mathscr{P})$ is uniquely determined by $\mathscr{P}$ and $\chi$. Each connected component $J^\alpha(\mathscr{P})$ is also an index category. \par
Below we denote $J(\mathscr{P})$ by $J$ for simplicity.
\begin{definition}\label{fdef}
$F^\bullet$ is a set of diagrams of $J^\bullet$ as $F^\bullet:=\{F^1, F^2,\cdots \}$, where $F^\alpha$ is a functor called a diagram of $J^\alpha$. $F^\alpha: J^\alpha \to \mathscr{P}$ is defined by
\begin{align*}
i\in ob(J^\alpha)\mapsto M_i\in ob(\mathscr{P}),\quad  f^k_{ij}\in J^\alpha(i,j)\mapsto T_{ij}^k\in \mathscr{P}(M_i,M_j)
\end{align*}
for all $i,j\in ob(J^\alpha)$ and $f^k_{ij}\in J^\alpha(i,j)$.
\end{definition}
Note that, since $J^\alpha$ is a connected category $F(J^\alpha)$ is also connected. However, objects which are connected in $\mathscr{P}$ are not always connected from the condition $({\rm ii})$ of Definition $\ref{jdef}$.
\begin{proposition}\label{limprop}
Let $\mathscr{P}$ be a pre-$\mathscr{Q}$ category. For arbitrary $P\in ob(\mathscr{P})$, if there is only one morphism from $\mathcal{A}$ to $P$, in other words if $\mathcal{A}$ is an initial object, then there exists a limit $M_\infty ^\alpha$ in $\mathscr{P}$ for each $(J^\alpha,~F^\alpha)$.
\end{proposition}
\begin{proof}
From the condition $2$ of Definition $\ref{preq}$ and the uniqueness of the morphism from $\mathcal{A}$, $\mathcal{A}$ is a candidate for the limit of all $(J^\alpha, F^\alpha)$. 
\end{proof}
The definition of the limit for $(J,F)$ used in this article can be found, for example, in Chapter $5$ in \cite{limit2} and \cite{limit1}. The limit $M_\infty$ is usually written as a combination with projections $\pi$ like $(M_\infty,\pi)$. We often denote $(M_\infty, \pi)$ as $M_\infty$ for short.\par
\bigskip
We can now define the quantization category.
\begin{definition}\label{qc}
A pre-$\mathscr{Q}$ category $\mathscr{P}$ is called a quantization category $\mathscr{Q}(\mathscr{P},J^\bullet,F^\bullet,\chi)$, when a limit $M_\infty^\alpha$ exists for each $(J^\alpha, F^\alpha)$ and the following conditions are satisfied at each limit $M_\infty ^\alpha$:\\
For arbitrary $f,g\in \mathcal{A}$ and $T:=T_\infty \in \mathscr{P}(\mathcal{A},M^\alpha_\infty)$, 
\begin{enumerate}
  \item[{\rm (Q1)}] $T(fg)-T(f)T(g)=0$,
  \item[{\rm (Q2)}] $\left[T(f),T(g)\right]_\infty-i\hbar(T)T(\{f,g\})=0$.
\end{enumerate}
\end{definition}
We denote $\mathscr{Q}(\mathscr{P},J^\bullet, F^\bullet, \chi)$ by $\mathscr{Q}$ for short. Note that all morphisms are linear maps since $\mathscr{Q}\subset R{\rm Mod}$. In general, therefore, a quantization map $T_i$ is not an algebraic homomorphism.

\bigskip

We show that this definition includes matrix regularization, strict deformation quantization, prequantization, and Poisson enveloping algebra in the following sections. 
\section{Matrix Regularization}\label{secmr}
In this section, we construct a quantization category to include matrix regularization of Definition $\ref{matrixreg}$.
\begin{definition}\label{cmrdef}
Let $\{N_i\}$ be a strictly increasing sequence of positive integers, and let $\hbar$ be a real-valued strictly positive decreasing function such that there exists positive real number $\delta \in \mathbb{R}$, $0<\lim_{N\to \infty}N\hbar(N)^\delta<\infty$. For a Poisson manifold $M$, a subcategory $\mathscr{P}_{MR}$ of $R${\rm Mod} is defined as follows:
\begin{align*}
ob(\mathscr{P}_{MR}):=\{ \mathcal{A}(M), Mat_{N_k}, Mat_{\infty}\mid k=1,2,\ldots \},
\end{align*}
where $\mathcal{A}(M)$ is $(C^\infty(M),\cdot,\{~,~\})$ on a Poisson manifold $M$, $Mat_{N_k}$ is a $N_k\times N_k$ matrix algebra and $Mat_{\infty}$ is a matrix algebra  which contains matrices of the limit of matrix regularization in Definition $\ref{matrixreg}$. The operator norm is defined for all $Mat_{N_i}$ and $Mat_\infty$. A morphism
\begin{align*}
T_i:&\mathcal{A}(M)\to Mat_{N_i}
\intertext{and a morphism}
T_\infty:&\mathcal{A}(M)\to Mat_\infty
\end{align*}
exist and they are quantization maps satisfying \ref{qcon1}., \ref{qcon2}. and \ref{qcon3}. in Definition $\ref{matrixreg}$ for all $Mat_{N_i}\in ob(\mathscr{P}_{MR})$. In addition, if and only if $N_i\le N_j$ $(N_i, N_j\in \mathbb{Z})$, morphisms $T_{ij}\in \mathscr{P}_{MR}(Mat_{N_j}, Mat_{N_i})$ satisfying
\begin{align}\label{eq1}
T_i=T_{ij}\circ T_{j},
\end{align}
exist. Let us define the set of morphisms $Mor(\mathscr{P}_{MR})$ by these quantization maps. We denote the set of quantization maps and identity maps by $QM_{MR}$, and the set of inverse maps of linear isomorphisms except for identities in $QM_{MR}$ by $QM^{-1}_{MR}$. That is,
\begin{align*}
QM_{MR}:&=\{T_i,T_{ij},T_\infty,id_{Mat_{N_i}},id_{Mat_{\infty}},id_{\mathcal{A}(M)}\mid i,j=1,2,\ldots \},\\
QM^{-1}_{MR}:&=\{f^{-1}\mid f\in QM_{MR}\mbox{ is an isomorphism and not an identity}\}.
\end{align*}
Then the set of morphisms $Mor(\mathscr{P}_{MR})$ is
\begin{align*}
Mor(\mathscr{P}_{MR}):=QM_{MR}\cup QM^{-1}_{MR}.
\end{align*}
For all morphisms $T_k\in \mathscr{P}_{MR}(\mathcal{A}(M),Mat_{N_k})~k=1,2,\ldots,$ and $T_\infty \in \mathscr{P}_{MR}(\mathcal{A}(M),Mat_{\infty})$, this codomain is equipped with the Lie bracket $[~,~]_k$ as the commutator such that 
\begin{align}\label{pmreq1}
[T_k(f),T_k(g)]_k=i\hbar(N_k)T_k(\{f,g\})+O(\hbar(N_k)^2)\quad (f,g\in \mathcal{A}(M)).
\end{align}
\end{definition}
This $\hbar$ goes to $0$ as $N\to \infty$ and controls the magnitude of noncommutativity.  The consistency of this definition is confirmed in the proofs of the following lemmas. As we will see later, it is enough to consider $QM^{-1}_{MR}$ is $\emptyset$ or $\{T^{-1}_\infty\}$.
\begin{lemma}\label{lem1}
When $Mat_{N_i}$ and $Mat_{N_j}$ are decomposed into ${\rm Im}~T_i\oplus {\rm Coker}~T_i$ and ${\rm Im}~T_j\oplus {\rm Coker}~T_j$, respectively, for all $T_{ij}\in \mathscr{P}_{MR}(Mat_{N_j},Mat_{N_i})$, $T_{ij}$ is block diagonal, i.e.,
\begin{align*}
T_{ij}=\left(
\begin{array}{c|c}
T_{ij}^{11}&0\\ \hline
0&T_{ij}^{22}
\end{array}\right):
\left(
\begin{array}{c}
{\rm Im}~T_j\\
{\rm Coker}~T_j
\end{array}\right)\to
\left(
\begin{array}{c}
{\rm Im}~T_i\\
{\rm Coker}~T_i
\end{array}\right).
\end{align*}
\end{lemma}
\begin{proof}
Let us denote $T_{ij}$ as 
\begin{align*}
T_{ij}=\left(
\begin{array}{c|c}
T_{ij}^{11}&T_{ij}^{12}\\ \hline
T_{ij}^{21}&T_{ij}^{22}
\end{array}\right):
\left(
\begin{array}{c}
{\rm Im}~T_j\\
{\rm Coker}~T_j
\end{array}\right)\to
\left(
\begin{array}{c}
{\rm Im}~T_i\\
{\rm Coker}~T_i
\end{array}\right).
\end{align*}
Since $T_{i}=T_{ij}\circ T_j$, $T^{12}_{ij}=T^{21}_{ij}=0$ is derived. 
\end{proof}
\begin{lemma}\label{cmrlem}
$\mathscr{P}_{MR}$ exist as a pre-$\mathscr{Q}$ category equipped with $T_{ij}\in \mathscr{P}_{MR}(Mat_{N_j},Mat_{N_i})$ for all $i,j~(i\le j)$.
\end{lemma}
\begin{proof}
From the definition of $\mathscr{P}_{MR}$, the conditions of Definition $\ref{preq}$ are satisfied, where $\hbar(T_k)=\hbar(N_k)$. We must show that there is no contradiction in composition maps of the two situations. For $\mathscr{P}_{MR}$ to be a category, 
\begin{itemize}
  \item the consistency $T_{ij}\circ T_{jk}=T_{ik}$ should be satisfied for all $i,j$ and $k$ such that $N_i\le N_j\le N_k$. 
  \item If $QM^{-1}_{MR}\neq \emptyset$, then $Mor(\mathscr{P}_{MR})$ includes all composition maps with $f^{-1}\in QM^{-1}_{MR}$.
\end{itemize}\par
First, we show that there is no contradiction even if $T_{ik}$ is simply defined as a composition of $T_{ij}\circ T_{jk}$. For adjacent objects $Mat_{N_i}$ and $Mat_{N_{i+1}}$, let us define $T_{i(i+1)}$ so that $T_{i(i+1)}T_{i+1}=T_i$ for every $i$. Since $N_i<N_{i+1}$, $T_{i(i+1)}^{11}$ in Lemma $\ref{lem1}$ exists. So the existence of this $T_{i(i+1)}$ for all $i$ is guaranteed. We define $T_{ij}$ by the ordered product as follows.
\begin{align}\label{eqlem}
T_{ij}:=T_{i(i+1)}T_{(i+1)(i+2)}\cdots T_{(j-1)j}.
\end{align}
These $T_{ij}$ trivially satisfy $T_{ij}\circ T_{jk}=T_{ik}$ for every $N_i\le N_j\le N_k$ at least if we put every $T_{ij}^{22}=0$. Now, we check that they do not contradict the condition $T_i=T_{ij}\circ T_j$. For the block matrix
\begin{align*}
T_{ij}=\left(
\begin{array}{c|c}
T_{ij}^{11}&0\\ \hline
0&T_{ij}^{22}
\end{array}\right),
\end{align*}
$\forall i,j$ and $\forall T^{22}_{ij}$, 
\begin{align*}
T_{ij}T_j&=T_{i(i+1)}T_{(i+1)(i+2)}\cdots T_{(j-2)(j-1)}T_{(j-1)(j)}T_j\\
&=T_{i(i+1)}T_{(i+1)(i+2)}\cdots T_{(j-2)(j-1)}T_{j-1}\\
&~~\vdots \\
&=T_i.
\end{align*}
Thus, $(\ref{eqlem})$ is no contradiction.\par
Next, we show that there is no contradiction even if $QM^{-1}_{MR}\neq \emptyset$. Since a dimension of arbitrary $Mat_{N_k}$ is finite and a dimension of $\mathcal{A}(M)$ is infinite, $\mathcal{A}(M)$ and $Mat_{N_k}$ are not isomorphic. So it is enough to consider the only case in which $T_\infty$ is an isomorphism. From the condition $(\ref{eq1})~T_k=T_{k\infty} \circ T_\infty$, and $T_\infty \circ T_\infty^{-1}=id_{Mat_{\infty}}$ derives
\begin{align*}
T_k\circ T_\infty^{-1}&=T_{k\infty}.
\end{align*}
Thus, $Mor(\mathscr{P}_{MR})$ includes all composition maps.
\end{proof}
From this proof, if $T_{ij}^{22} T_{jk}^{22}=T_{ik}^{22}$ is satisfied, then $\mathscr{P}_{MR}$ is a consistently pre-$\mathscr{Q}$ category.
These morphisms $T_i,~T_{ij}$ for $i,j=1,2,\cdots $ can be specifically configured in the case of Berezin-Toeplitz quantization as follows.
\bigskip
\bigskip
\begin{example}[\cite{Le}]\label{btexample}
Let $(M,\omega )$ be a compact connected K\"{a}hler manifold with a line bundle $L\to M$. We denote $L^{\otimes k}$ by $L^k$ for all $k\ge 1$, and a smooth sections of $L^k$ by $C^\infty(M,L^k)$. The quantum space at $k$ given as
\begin{align*}
\mathcal{H}_k=H^0(M,L^k)
\end{align*}
of global holomorphic sections of $L^k\to M$. Here, $H^0(M,L^k)$ is Dolbeault cohomology. Let $L^2(M, L^k)$ be the completion of $C^\infty(M,L^k)$ with respect to the inner product $\langle\cdot,~\cdot\rangle_k$, and $\Pi_k$ be the orthogonal projector from $L^2(M,L^k)$ to $\mathcal{H}_k$. The Berezin-Toeplitz operator is defined as 
\begin{align*}
T_k(f)=\Pi_k\circ \mu_f\circ \Pi_k,
\end{align*}
where $\mu_f$ is the operator of multiplication by $f\in C^\infty(M)$. That is, $T_k$ is a linear map from the commutative algebra $C^\infty(M)$ to a matrix algebra ${\rm End}(\mathcal{H}_k)$.
\begin{theorem}[\cite{berezin2}]\label{thmbtdim}
The dimension of the $\mathcal{H}_k$ satisfies
\begin{align*}
\dim \mathcal{H}_k=\left(\frac{k}{2\pi}\right)^n{\rm vol}(M)+O(k^{n-1})
\end{align*}
when $k$ goes to infinity, where $2n=\dim M$.
\end{theorem}
From this theorem, $T_{ij}$ is constructed as follows. Let $|n\rangle _k$ be an orthonormal basis that diagonalizes $\Pi_k$, i.e., $\Pi_k=\sum_n^{\dim \mathcal{H}_k}|n\rangle_k {}_k\langle n|$. For $\dim \mathcal{H}_j\ge \dim \mathcal{H}_i$,
\begin{align*}
T_{ij}T_j(f):&=T_{ij}(\Pi_j\circ \mu_f\circ \Pi_j)\\
&=U_{ij}\Pi_jU^\dagger _{ij} \circ \mu_f\circ U_{ij} \Pi_j U_{ij}^\dagger\\
&=\Pi_i\circ \mu_f \circ \Pi_i\\
&=T_i(f),
\end{align*}
where $U_{ij}=\sum_n^{\dim \mathcal{H}_i} |n\rangle_i{}_j\langle n|$.
\end{example}
For the Berezin-Toeplitz quantization map $T_k$, the following theorem is known.
\begin{theorem}[\cite{berezin2}]\label{thmbt}
Let $(M,\omega )$ be a compact, connected, K\"{a}hler manifold and $T_k$ be the Berezin-Toeplitz operator.
\begin{enumerate}
  \item[(a)] For any $f\in C^\infty(M)$, 
\begin{align*}
\|T_k(f)\|\le \|f\|_\infty,
\end{align*}
where $\|T\|$ stands for the operator norm and $\|T\|_\infty$ stands for the uniform norm.
 \item[(b)] For any $f,g\in C^\infty(M)$,
\begin{align*}
T_k(f)T_k(g)=T_k(fg)+O(k^{-1}).
\end{align*}
 \item[(c)] For any $f,g\in C^\infty(M)$,
\begin{align}\label{btc}
[T_k(f),T_k(g)]=\frac{1}{ik}T_k(\{f,g\})+O(k^{-2}).
\end{align}
\end{enumerate}
\end{theorem}
Thus, the Berezin-Toeplitz quantization map $T_k$ is the homomorphism of algebra between $C^\infty(M)$ and ${\rm End}(\mathcal{H}_k)$ when $k\to \infty$. Comparing $(\ref{pmreq1})$ and $(\ref{btc})$, to regard Berezin-Toeplitz quantization as the category of Definition $\ref{cmrdef}$, $\hbar$ should be defined by $\hbar=-1\slash k$. In addition, from Theorem $\ref{thmbtdim}$, we find that $\delta$ in Definition $\ref{cmrdef}$ is given by $\delta=n$ when $\dim M=2n$. Therefore, Berezin-Toeplitz quantization is an example of $\mathscr{P}_{MR}$.

\bigskip

\bigskip
$\mathscr{P}_{MR}(\mathcal{A}(M),Mat_{N_i})$ has only one quantization map $T_i$ by the definition. In such a case, $\chi(Mat_{N_i})=\hbar(T_i)$. Therefore the character $\chi $ is obtained by $\chi(Mat_{N_i})=\hbar(T_i):=\hbar(N_i)$. This $\chi$ gives the index category $J^\bullet_{MR}$ and the set of diagrams $F^\bullet_{MR}$ as follows.  In the case of $\mathscr{P}_{MR}$, $J^\bullet_{MR}$ has only one connected component $J_{MR}$, i.e., $J^\bullet_{MR}=\{J_{MR}\}$.  So, $J_{MR}$ is given as a directed set such that there exists a morphism $j\to i$ if and only if $i\le j$ from Lemma $\ref{cmrlem}$. For $F^\bullet_{MR}:=\{F_{MR}\}$, a diagram $F_{MR}$ of $J_{MR}$ is given by 
\begin{align*}
\begin{array}{rccc}
F_{MR}:&ob(J_{MR})&\to &ob(\mathscr{P}_{MR})\\
{}&\rotatebox{90}{$\in$}&{}&\rotatebox{90}{$\in$}\\
{}&i&\mapsto &Mat_{N_i}
\end{array}, \quad
\begin{array}{rccc}
F_{MR}:&J_{MR}(i,j)&\to &\mathscr{P}_{MR}(Mat_{N_i},Mat_{N_j})\\
{}&\rotatebox{90}{$\in$}&{}&\rotatebox{90}{$\in$}\\
{}&f^k_{ij}&\mapsto &T_{ij}^k
\end{array}.
\end{align*} 
Since $\{N_{i}\}$ is a strictly increasing sequence, $F_{MR}$ is a diagram which satisfies Definition $\ref{fdef}$ from a sequence in $J_{MR}$ 
\begin{align*}
\xymatrix{
\cdots \ar[r]& k\ar[r] \ar @/_12pt/[rr] & j\ar[r]& i \ar[r]& \cdots
}
\end{align*}
to
\begin{align*}
\xymatrix{
\cdots \ar[r]& Mat_{N_k}\ar[r]^{T_{jk}} \ar @/_18pt/[rr]_{T_{ik}} & Mat_{N_j}\ar[r]^{T_{ij}}& Mat_{N_i} \ar[r]& \cdots\\
{}&
}
\end{align*}
in $\mathscr{P}_{MR}$, where $k\ge j\ge i$.
\begin{theorem}\label{qmrthm}
For $(\mathscr{P}_{MR}$, $J_{MR}^\bullet$, $F^\bullet_{MR}$, $\chi)$  given as above, $\mathscr{Q}_{MR}:=\mathscr{Q}(\mathscr{P}_{MR}, J^\bullet_{MR}, F^\bullet_{MR}, \chi)$ is a quantization category of the matrix regularization with the limit $Mat_{\infty}$ of $(J_{MR}, F_{MR})$.
\end{theorem}
\begin{proof}
From the Lemma $\ref{cmrlem}$, $\mathscr{P}_{MR}$ exists as a pre-$\mathscr{Q}_{MR}$ category. A limit $(M_\infty, \pi)$ of $F_{MR}$ is given as the following commutative diagram for any $i,j$. 
\begin{align*}
\xymatrix{
{}&X \ar@{.>}[d]^T \ar[ld] \ar[rd]\\
Mat_{N_i}&M_\infty \ar[l]^{~~\pi_i} \ar[r]_{\pi_j~~}&Mat_{N_j} \ar @/^17pt /[ll]^{T_{ij}}
}
\end{align*}
From the condition $(\ref{eq1})$ of morphisms, $M_\infty$ is $\mathcal{A}(M)$ or $Mat_\infty$. Lemma $\ref{cmrlem}$ shows that $M_\infty=Mat_\infty$, i.e., $T=T_\infty$, $\pi_i=T_{i\infty}$ and $\pi_j=T_{j\infty}$, and then $\pi_i\circ T=T_i$ and $T_{ij}\circ \pi_j=\pi_i$ for all $i,j$ with $X=\mathcal{A}(M)$. From the definition of $Mat_\infty$, the quantization conditions $(Q1)$ and $(Q2)$ in Definition $\ref{qc}$ are satisfied on $M_\infty$. When $Mat_\infty \simeq \mathcal{A}(M)$, $\mathcal{A}(M)$ is also the limit, and the quantization conditions $(Q1)$ and $(Q2)$ in Definition $\ref{qc}$ are trivially satisfied on $\mathcal{A}(M)$. Thus $\mathscr{Q}_{MR}$ is a quantization category of the matrix regularization.
\end{proof}
\begin{example}
The Berezin-Toeplitz quantization in Example $\ref{btexample}$ is an example of $\mathscr{Q}_{MR}$.
\end{example}
\bigskip
This quantization category $\mathscr{Q}_{MR}$ can be further minimized.
\begin{corollary}\label{qmrcor}
For the quantization category $\mathscr{Q}_{MR}$, a restriction $\left.\mathscr{Q}_{MR}\vphantom{\big|} \right|_{Mat_{N_i}}$ is given as follows:
\begin{align*}
ob\left(\left.\mathscr{Q}_{MR}\vphantom{\big|} \right|_{Mat_{N_i}}\right):=\{\mathcal{A}(M),Mat_{N_i}, Mat_{\infty }\},
\end{align*}
where $i$ is fixed. Morphisms are restricted to $T_i,T_\infty, T_{i\infty}, id_{\mathcal{A}(M)}, id_{Mat_{N_i}}, id_{Mat_{\infty}}$, and $T^{-1}_\infty$ if $T_\infty$ is an isomorphism. Then $\left.\mathscr{Q}_{MR}\vphantom{\big|} \right|_{Mat_{N_i}}$ is a quantization category of the matrix regularization with the limit $Mat_{\infty}$ of $(J_{MR}|_i, F_{MR}|_i)$. Here, $J_{MR}|_i$ is given by $ob(J)=\{i,\infty\}$ and the only morphism of $J_{MR}$ is $i\to \infty$, and $F_{MR}|_i$ is given by $F(i)=Mat_{N_i}$, $F(\infty)=Mat_{\infty}$, $F(i\to \infty)=T_{i\infty}$.
\end{corollary}
\begin{proof}
Recall that if $QM_{DQ}^{-1}\neq \emptyset$, then there only exists $T^{-1}=T^{-1}_\infty \in QM_{MR}^{-1}$, because the dimension of $Mat_{N_i}$ is finite, but the dimension of $\mathcal{A}(M)$ is infinite, as we saw in the proof for Lemma \ref{cmrlem}. Since the diagram of $\left.\mathscr{Q}_{MR}\vphantom{\big|} \right|_{Mat_{N_i}}$ is given by
\begin{align}
\xymatrix{
\mathcal{A}(M) \ar@{.>}[d]^T \ar[rd]\\
M_\infty \ar[r]_{\pi_i}&Mat_{N_i} ,
}\label{diagrammr}
\end{align}
this corollary is derived in the same way of Theorem $\ref{qmrthm}$.
\end{proof}
\section{Deformation Quantization}\label{secdq}
In the following, we consider the strict deformation quantization. In other words, we consider not $(\mathcal{F},*)$ but $(\mathcal{A}_0(M),\mathcal{Q}^\hbar)$. (These examples are shown in \cite{rieffel1,rieffel2}.)
\begin{definition}\label{cdqdef}
Let $(\mathcal{A}_0(M),\mathcal{Q}^\hbar)$ be the strict deformation quantization of a Poisson manifold $M~($see Definitions $\ref{strict1}$ and $\ref{strict2})$. $\mathcal{C}^\hbar$ and $\mathcal{C}^\hbar_{\mathbb{R}}$ are those in Definition $\ref{strict1}$. A subcategory $\mathscr{P}_{DQ}$ of $R${\rm Mod} is defined as follows: 
\begin{align*}
ob(\mathscr{P}_{DQ}):=\{\mathcal{A}_0(M), \mathcal{C}^\hbar \mid \forall\hbar \in I \},
\end{align*}
where $I$ is a subset of real numbers that contains $0$. The morphisms of $\mathscr{P}_{DQ}$ are defined by quantization maps $\mathcal{Q}^\hbar:\mathcal{A}_0(M)\to \mathcal{C}^\hbar_{\mathbb{R}}\subset \mathcal{C}^\hbar$ for all $\hbar \in I$. In addition, if and only if $\hbar\ge \hbar^\prime$, $T_{\hbar\hbar^\prime}\in \mathscr{P}_{DQ}(\mathcal{C}^{\hbar^\prime},\mathcal{C}^\hbar)$ satisfying
\begin{align*}
\mathcal{Q}^\hbar=T_{\hbar\hbar^\prime}\circ \mathcal{Q}^{\hbar^\prime}
\end{align*}
exist. Let us define the set of morphisms $Mor(\mathscr{P}_{DQ})$ by these quantization maps. We denote the set of quantization maps and identity maps by $QM_{DQ}$, and the set of inverse maps of linear isomorphisms except for identities in $QM_{DQ}$ by $QM^{-1}_{DQ}$. That is,
\begin{align*}
QM_{DQ}:&=\{Q_\hbar, T_{\hbar \hbar^\prime}, id_{\mathcal{A}_{0}(M)}, id_{\mathcal{C}^\hbar}\mid \hbar ,\hbar^\prime \in I \},\\
QM^{-1}_{DQ}:&=\{f^{-1}\mid f\in QM_{DQ}\mbox{ is an isomorphism and not an identity}\},
\end{align*}
such that $T_{\hbar \hbar^\prime}=id$ when $\hbar=\hbar^\prime$. The set of other morphisms $Comp_{DQ}$ is defined by composition maps such that $\mathscr{P}_{DQ}$ is a category and satisfies the condition $5$ in Definition \ref{preq}. Then the set of morphisms $Mor(\mathscr{P}_{DQ})$ is defined by
\begin{align*}
Mor(\mathscr{P}_{DQ}):=QM_{DQ}\cup QM^{-1}_{DQ} \cup Comp_{DQ}.
\end{align*}
For all morphisms $\mathcal{Q}^\hbar\in \mathscr{P}_{DQ}(\mathcal{A}_0(M),\mathcal{C}^\hbar)$, the codomain is equipped with the Lie bracket $[~,~]_\hbar$ as the commutator such that 
\begin{align*}
[\mathcal{Q}^\hbar(f),\mathcal{Q}^\hbar(g)]_\hbar= \sqrt{-1}\hbar \mathcal{Q}^\hbar(\{f,g\})+O(\hbar ^2)\quad (f,g\in \mathcal{A}_0(M)).
\end{align*}
\end{definition}
\bigskip
From the definition of $\mathscr{P}_{DQ}$, the following lemma is derived immediately.
\begin{lemma}\label{cdqlem}
$\mathscr{P}_{DQ}$ is a pre-$\mathscr{Q}$ category.
\end{lemma}
%
%
\par
The character $\chi$ is obtained by $\chi(\mathcal{C}^\hbar)=|\hbar|$. This character $\chi$ gives the index category $J^\bullet_{DQ}$ and the set of diagrams $F^\bullet_{DQ}$ as follows. For $J^\bullet_{DQ}:=\{J_{DQ}\}$, $J_{DQ}$ is given as a directed set. For $F^\bullet_{DQ}:=\{F_{DQ}\}$, a diagram $F_{DQ}$ of $J_{DQ}$ is surjective from $ob(J_{DQ})\to ob(\mathscr{P}_{DQ})\backslash \{\mathcal{A}_0(M)\}$. Since $\chi(\mathcal{C}^\hbar)=|\hbar|$ for all $\hbar$, if and only if $|\hbar |\ge |\hbar^\prime|$ there exist morphisms $\hbar^\prime\in ob(J_{DQ}) \mapsto \hbar\in ob(J_{DQ})$ such that
\begin{align*}
\begin{array}{rccc}
F_{DQ}:&ob(J_{DQ})&\to &ob(\mathscr{P}_{DQ})\\
{}&\rotatebox{90}{$\in$}&{}&\rotatebox{90}{$\in$}\\
{}&\hbar&\mapsto &\mathcal{C}^\hbar
\end{array}, \quad
\begin{array}{rccc}
F_{DQ}:&J_{DQ}(\hbar^\prime,\hbar )&\to &\mathscr{P}_{DQ}(\mathcal{C}^{\hbar^\prime},\mathcal{C}^\hbar)\\
{}&\rotatebox{90}{$\in$}&{}&\rotatebox{90}{$\in$}\\
{}&f^k_{\hbar \hbar^\prime}&\mapsto &T_{\hbar \hbar^\prime}^k
\end{array}.
\end{align*}
We will define $\tilde{\mathscr{P}}_{DQ}$ for the case that there is an only isomorphism $(\mathcal{Q}^0)^{-1}\in QM_{DQ}^{-1}$ or there is no isomorphism $QM^{-1}_{DQ}=\emptyset$.
\begin{definition}
For the pre-$\mathscr{Q}$ category of the strict deformation quantization $\mathscr{P}_{DQ}$, $\tilde{\mathscr{P}}_{DQ}$ is defined as a $\mathscr{P}_{DQ}$ satisfying
\begin{align*}
QM^{-1}_{DQ}=\emptyset \quad \mbox{or}\quad QM^{-1}_{DQ}=\{(\mathcal{Q}^0)^{-1}\}.
\end{align*}
\end{definition}
\begin{theorem}\label{qdqthm}
For $(\tilde{\mathscr{P}}_{DQ},J^\bullet_{DQ},F^\bullet_{DQ},\chi)$ given as above, $\tilde{\mathscr{Q}}_{DQ}:=\mathscr{Q}(\tilde{\mathscr{P}}_{DQ},J^\bullet_{DQ},F^\bullet_{DQ},\chi)$ is a quantization category $\mathscr{Q}_{DQ}$ of the strict deformation quantization with a limit $\mathcal{C}^0$ of $(J_{DQ},F_{DQ})$.
\end{theorem}
\begin{proof}
The proof is given in a completely parallel manner as the proof of Theorem $\ref{qmrthm}$. From Lemma $\ref{cdqlem}$, $\tilde{\mathscr{P}}_{DQ}$ is a pre-$\mathscr{Q}$ category. We consider a limit $M_\infty$ of $F_{DQ}$ as the following commutative diagram.
\begin{align*}
\xymatrix{
{}&\mathcal{A}_0(M) \ar@{.>}[d]^{\mathcal{Q}^0} \ar[ld]_{\mathcal{Q}^\hbar} \ar[rd]^{\mathcal{Q}^{\hbar^\prime}}\\
\mathcal{C}^\hbar&M_\infty=\mathcal{C}^0\ar[l]^{T_{\hbar 0}~~} \ar[r]_{~~T_{\hbar^\prime 0}}&\mathcal{C}^{\hbar^\prime}\ar @/^19pt /[ll]^{T_{\hbar\hbar^\prime}}
}
\end{align*}
Thus, the limit of $F_{DQ}$ is $\mathcal{C}^0$. Trivially, the quantization conditions are satisfied on $\mathcal{C}^0$, so $\tilde{\mathscr{Q}}_{DQ}$ is a quantization category of the strict deformation quantization.
\end{proof}
The quantization category $\tilde{\mathscr{Q}}_{DQ}$ can be further minimized in a manner similar to $\mathscr{Q}_{MR}$.
\begin{corollary}\label{qdqcor}
For the quantization category $\tilde{\mathscr{Q}}_{DQ}$, a restriction $\left.\tilde{\mathscr{Q}}_{DQ}\vphantom{\big|} \right|_{\mathcal{C}^\hbar}$ is given as follows:
\begin{align*}
ob\left(\left.\tilde{\mathscr{Q}}_{DQ}\vphantom{\big|} \right|_{\mathcal{C}^\hbar}\right):=\{\mathcal{C}^0,\mathcal{C}^\hbar, \mathcal{A}_0(M)\},
\end{align*}
where $\hbar$ is fixed. Similarly, the morphisms are restricted to $\mathcal{Q}^0, \mathcal{Q}^\hbar, T_{\hbar 0}$. 
\begin{align}
\xymatrix{
\mathcal{A}_0(M) \ar@{.>}[d]^{\mathcal{Q}^0} \ar[rd]^{\mathcal{Q}^{\hbar}}  \\
M_\infty=\mathcal{C}^0 \ar[r]_{~~T_{\hbar 0}}&\mathcal{C}^{\hbar}
}\label{diagramdq}
\end{align}
Then $\left.\tilde{\mathscr{Q}}_{DQ}\vphantom{\big|} \right|_{\mathcal{C}^\hbar}$ is a quantization category of the strict deformation quantization with the limit $\mathcal{C}^0$. 
\end{corollary}
We showed that at least $\tilde{\mathscr{Q}}_{DQ}$ is a quantization category. However, the following case is not a quantization category.
\begin{theorem}\label{theoQM-1}
For a pre-$\mathscr{Q}$ category $\mathscr{P}_{DQ}$, if there exists an isomorphism $(\mathcal{Q}^{\hbar})^{-1}\in QM_{DQ}^{-1}$ for $\hbar\neq 0$, and $\mathcal{Q}^0\circ (\mathcal{Q}^{\hbar})^{-1}\circ T_{\hbar 0}\neq id_{\mathcal{C}^0}$, then $\mathscr{Q}_{DQ}:=\mathscr{Q}(\mathscr{P}_{DQ},J^\bullet_{DQ},F^\bullet_{DQ},\chi)$ is not a quantization category.
\end{theorem}
\begin{proof}
We consider the case that there exists $(\mathcal{Q}^\hbar)^{-1}$for $\hbar \neq 0$ such that $\mathcal{Q}^\hbar \circ (\mathcal{Q}^{\hbar})^{-1}=id_{\mathcal{C}^\hbar}, (\mathcal{Q}^\hbar)^{-1}\circ \mathcal{Q}^\hbar=id_{\mathcal{A}_0(M)}$.
\begin{align*}
\xymatrix{
\mathcal{A}_0(M) \ar@{.>}[d]_{\mathcal{Q}^0} \ar@<-0.4ex>[rd]_{\mathcal{Q}^{\hbar}}  \\
\mathcal{C}^0 \ar[r]_{T_{\hbar 0}}&\mathcal{C}^{\hbar} \ar@<-0.4ex>[ul]_{(\mathcal{Q}^{\hbar})^{-1}}
}
\end{align*}
We show the proof by cases:
\begin{enumerate}
  \item $\mathcal{Q}^0$ is an isomorphism.
  \item $\mathcal{Q}^0$ is not an isomorphism.
\end{enumerate}
First, let $\mathcal{Q}^0$ be an isomorphism. Then $\mathcal{A}_0(M)\simeq \mathcal{C}^0\simeq \mathcal{C}^\hbar$, and $\mathcal{A}_0(M)$, $\mathcal{C}^0$, and $\mathcal{C}^\hbar$ are the limit. However, since $\hbar\neq 0$ for $\mathcal{C}^\hbar$ and $\mathcal{A}_0(M)$ is a commutatve algebra, if the quantization condition $(Q1)$ in Definition $\ref{qc}$ is satisfied, the quantization condition $(Q2)$ is not satisfied. Thus, this case is not a quantization category. \par
Next, we assume $\mathcal{Q}^0$ be not an isomorphism. As $e_B$  in Proposition \ref{apeB3}, there exists a morphism $\mathcal{Q}^0\circ (\mathcal{Q}^{\hbar})^{-1}\circ T_{\hbar 0}\in \mathscr{Q}_{DQ}(\mathcal{C}^0,\mathcal{C}^0)$ when $\mathcal{Q}^\hbar$ is an isomorphism. Then the following diagram is commuted by $\mathcal{Q}^0\circ (\mathcal{Q}^{\hbar})^{-1}\circ T_{\hbar 0}$ and $id_{\mathcal{C}^0}$,
\begin{align*}
\xymatrix{
\mathcal{C}^0 \ar@<-0.5ex>[d]_{\mathcal{Q}^0\circ (\mathcal{Q}^{\hbar})^{-1}\circ T_{\hbar 0}}\ar@<0.5ex>[d]^{id} \ar[rd]^{T_{\hbar 0}}  \\
\mathcal{C}^0 \ar[r]_{T_{\hbar 0}}&\mathcal{C}^{\hbar}.
}
\end{align*}
Since the morphism $\mathcal{C}^0\to \mathcal{C}^0$ that commutes this diagram is not unique, the limit is not $\mathcal{C}^0$. If the limit is defined, the limit satisfying the quantization conditions $(Q1)$ and $(Q2)$ in Definition $\ref{qc}$ is $\mathcal{A}_0(M)$. However, since $\mathcal{A}_0(M)\simeq \mathcal{C}^\hbar$, $\mathcal{C}^\hbar$ is also the limit when $\mathcal{A}_0(M)$ is the limit. Thus, this case is not a quantization category as in the first case.
\end{proof}
\section{Prequantization}\label{secpq}
It is possible to construct a quantization category including the prequantization whose definition is given in Definition $\ref{def_prequantization}$.
\begin{definition}\label{cpqdef}
Let $(M,~\{~,~\})$ be a Poisson manifold with a prequantum line bundle $L\to M$ and $\Gamma_{hol}(M,L)$ be a holomorphic global section on $L$. Let $\mathcal{A}(M)$ be a Poisson algebra $(C^\infty(M), \cdot, \{~,~\})$ for the Poisson manifold. A subcategory $\mathscr{P}_{PQ}$ of $R${\rm Mod} is defined as follows. The set of objects is :
\begin{align*}
ob(\mathscr{P}_{PQ}):=\{\mathcal{A}(M), {\rm End}(\Gamma_{hol}(M,L)), \mathcal{T}_{\mathcal{A}(M)}\},
\end{align*}
where $\mathcal{T}_{\mathcal{A}(M)}$ is the set of multiplication operators by arbitrary functions in $\mathcal{A}(M)$ acting on $\Gamma_{hol}(M,L)$. We denote ${\rm End}(\Gamma_{hol}(M,L))$ by $End$ for simplicity. Morphisms of $\mathscr{P}_{PQ}$ are given by identity maps and following maps:
\begin{align*}
\begin{array}{rccc}
P:&\mathcal{A}(M)&\to &End\\
{}&\rotatebox{90}{$\in$}&{}&\rotatebox{90}{$\in$}\\
{}&f&\mapsto&P(f)
\end{array},
\begin{array}{rccc}
\iota:&\mathcal{A}(M)&\to &\mathcal{T}_{\mathcal{A}(M)}\\
{}&\rotatebox{90}{$\in$}&{}&\rotatebox{90}{$\in$}\\
{}&f&\mapsto&\mu_f .
\end{array} 
\end{align*}
Here $P(f):=\hat{f}$ in Definition $\ref{def_prequantization}$, and $\mu_f$ is an operator of multiplication by $f\in \mathcal{A}(M)$. In addition, there is a map $T\in \mathscr{P}_{PQ}(\mathcal{T}_{\mathcal{A}(M)},End)$ satisfying 
\begin{align*}
P=T\circ \iota.
\end{align*}
We denote the set of quantization maps and identity maps by $QM_{PQ}$, and the set of inverse $\iota$ in $QM_{PQ}$ by $QM^{-1}_{PQ}$. That is,
\begin{align*}
QM_{PQ}:&=\{P, \iota, T,id_{\mathcal{A}(M)}, id_{End}, id_{\mathcal{T}_{\mathcal{A}(M)}}\},\\
QM^{-1}_{PQ}:&=\{\iota^{-1}\}.
\end{align*}
Then the set of morphisms $Mor(\mathscr{P}_{PQ})$ is
\begin{align*}
Mor(\mathscr{P}_{PQ}):=QM_{PQ}\cup QM^{-1}_{PQ},
\end{align*}
For $P\in\mathscr{P}_{PQ}(\mathcal{A}(M), End)$, the codomain is equipped with the Lie bracket $[~,~]$ as the commutator such that 
\begin{align*}
[P(f),P(g)]=i\hbar P(\{f,g\}).
\end{align*}
\end{definition}
\bigskip

Note that a morphism $\iota$ is clearly an isomorphism between $\mathcal{A}(M)$ and $\mathcal{T}_{\mathcal{A}(M)}$. So there exists $\iota ^{-1}\in QM^{-1}_{PQ}$. For $\iota^{-1}$, its a composition map is given by
\begin{align}
P\circ\iota^{-1}=T\circ \iota \circ \iota^{-1}=T. \label{eqpreq}
\end{align}
Thus $\mathscr{P}_{PQ}$ is a category. From the definition of $\mathscr{P}_{PQ}$, the following lemma is derived immediately.
\begin{lemma}\label{cpqlem}
$\mathscr{P}_{PQ}$ is a pre-$\mathscr{Q}$ category.
\end{lemma}
%
%
%
%
The character $\chi$ is obtained by $\chi(End)=|\hbar|$ and $\chi(\mathcal{T}_{\mathcal{A}(M)})=0$. A diagram $F_{PQ}^\bullet$ maps from
\begin{align*}
\xymatrix{
0\ar[r]& 1
}
\end{align*}
to
\begin{align*}
\xymatrix{
\mathcal{T}_{\mathcal{A}(M)}\ar[r] &End.
}
\end{align*}
\begin{theorem}\label{qpqthm}
For $(\mathscr{P}_{PQ},J^\bullet_{PQ},F^\bullet_{PQ},\chi)$ given as above, $\mathscr{Q}_{PQ}:=\mathscr{Q}(\mathscr{P}_{PQ},J^\bullet_{PQ},F^\bullet_{PQ},\chi)$ is a quantization category $\mathscr{Q}_{PQ}$ of the prequantization with the limit $\mathcal{T}_{\mathcal{A}(M)}$ of $(J_{PQ},F_{PQ})$.
\end{theorem}
\begin{proof}
From Lemma $\ref{cpqlem}$, $\mathscr{Q}_{PQ}$ is a pre-$\mathscr{Q}$ category. We consider a limit $M_\infty$ of $F_{PQ}$ as the following commutative diagram.
\begin{align}\label{diagrampq}
\xymatrix{
\mathcal{A}(M) \ar@{.>}[d]^{\iota} \ar[rd]^{P}\\
M_\infty=\mathcal{T}_{\mathcal{A}(M)} \ar[r]_{T}&End 
}
\end{align}
Thus, the limit of $F_{PQ}$ is $\mathcal{T}_{\mathcal{A}(M)}$. The quantization conditions are satisfied on $\mathcal{T}_{\mathcal{A}(M)}$, so $\mathscr{Q}_{PQ}$ is a quantization category of the pre-quantization.
\end{proof}
We comment the reason why the $QM_{PQ}^{-1}$ is restricted to $\{\iota^{-1}\}$. If there exists $P^{-1}\in QM_{PQ}^{-1}$, then $End$ is also the limit. So $\mathscr{Q}_{PQ}$ is not a quantization category as similar to Theorem $\ref{theoQM-1}$.
\section{Poisson Enveloping Algebra}\label{secenv}
In this section, we describe the relationship between the quantization category and Poisson enveloping algebra. 

First, we review the definition of the Poisson enveloping algebra.
\begin{definition}[\cite{oh,oh2}] \label{def.env}
Let $\mathcal{A}=(\mathcal{A},\cdot,\{~,~\})$ be a Poisson algebra. For an algebra $U$, let $\alpha _0$ be an algebra homomorphism from $(\mathcal{A},\cdot)$ to $U$ and $\beta _0$ be a Lie homomorphism from $(\mathcal{A}, \{~,~\})$ to $U$. These morphisms satisfy the following conditions for all $a,b\in \mathcal{A}$.
\begin{align}
\alpha_0(\{a,b\})&=\beta_0(a)\alpha_0(b)-\alpha_0(b)\beta_0(a),\label{eq.alpha}\\
\beta_0(ab)&=\alpha_0(a)\beta_0(b)+\alpha_0(b)\beta_0(a).\label{eq.beta}
\end{align}
Let $X_k$ be an arbitrary algebra such that there exist an algebra homomorphism $\alpha_k$ from $(\mathcal{A},\cdot)$ to $X_k$ and a Lie homomorphism $\beta_k$ from $(\mathcal{A},\{~,~\})$ to $X_k$ such that 
\begin{align}
\alpha_k(\{a,b\})&=\beta_k(a)\alpha_k(b)-\alpha_k(b)\beta_k(a),\label{eq.gamma}\\
\beta_k(ab)&=\alpha_k(a)\beta_k(b)+\alpha_k(b)\beta_k(a).\label{eq.delta}
\end{align}
If there exists a unique algebra homomorphism $h_k$ from $U$ to $X_k$ that makes the following diagram $(\ref{eq.env})$ commutative, then $(U,\alpha_0,\beta_0)$ is called a Poisson enveloping algebra of $\mathcal{A}$.
\begin{align}
\xymatrix{
{}&X_k \\
(\mathcal{A},\cdot) \ar[ru]^{\alpha_k }\ar[r]_{\alpha_0 } & U\ar@{.>}[u]_{h_k}&(\mathcal{A},\{~,~\}) \ar[lu]_{\beta_k }\ar[l]^{\beta_0 } 
 } \label{eq.env}
\end{align}

\end{definition}
\begin{example}[\cite{um}] \label{thm.um}
A Poisson symplectic algebra $P_n$ is given by the polynomial algebra $R[x_1,\ldots,x_n,y_1,\ldots,y_n]$ that is equipped with the Poisson bracket defined by 
\begin{align*}
\{x_i,y_j\}=\delta_{ij},\quad
\{x_i,x_j\}=0, \quad \{y_i,y_j\}=0,
\end{align*}
where $1\le i,j\le n\in \mathbb{Z}$. Let $A_{2n}$ be the Weyl algebra such that $A_{2n}$ is an associative algebra given by generators $x_1,\ldots,x_{2n}, y_1,\ldots, y_{2n}$ and defined relations
\begin{align*}
[x_i,y_j]=\delta_{ij},\quad [x_i,x_j]=0,\quad [y_i,y_j]=0,
\end{align*}
where $1\le i,j\le 2n\in \mathbb{Z}$. Then, the enveloping algebra $P_n^e$ of $P_n$ and the Weyl algebra $A_{2n}$ are isomorphic.
\end{example}
From Example $\ref{thm.um}$ and so on, we are convinced that the Poisson enveloping algebra is a kind of quantization. 
\begin{definition}\label{cenvdef}
Let $U:=(U,\alpha_0 ,\beta_0)$ be the Poisson enveloping algebra with a Poisson algebra $\mathcal{A}=(\mathcal{A},\cdot~,\{~,~\})$. A subcategory $\mathscr{P}_{env}$ of $R${\rm Mod} is defined as follows:
\begin{align*}
ob(\mathscr{P}_{env}):=\{\mathcal{A}, (\mathcal{A}, \cdot ), U, M_1, M_2, \ldots \},
\end{align*} 
where $M_i$ for $i=1,2,\cdots ,$ are Lie algebras such that the following diagram commutes.
\begin{align*}
\xymatrix{
(\mathcal{A},\cdot)\ar[rd]_{\alpha_0} \ar[rdd]_{\alpha_i}&\ar[l]_{~~\Pi_1}\mathcal{A}\ar[r]^{\Pi_2~~}\ar[d]_{\alpha_0 \circ \Pi_1} \ar[d]^{\beta_0 \circ \Pi_2}&\ar[ldd]^{\beta_i} \ar[ld]^{\beta_0}(\mathcal{A},\cdot, \{~,~\})\\
{}&U\ar@{.>}[d]^{h_{i}}&{}\\
{}&M_i&{}
}
\end{align*}
Morphisms of $\mathscr{P}_{env}$ are given by 
\begin{align*}
\begin{array}{rccccrccccc}
\Pi_1:&\mathcal{A}&\to &(\mathcal{A},\cdot)&,~&\Pi_2:&\mathcal{A}&\to &(\mathcal{A},\cdot,\{~,~\}),\\
\end{array}
\end{align*}
where $\Pi_1$ is the identity forgetting the $\{~,~\}$ structure and $\Pi_2$ is the identity, and the homomorphisms and Lie homomorphisms are as follows:
\begin{align*}
\alpha_0:&(\mathcal{A},\cdot)\to U,\\
\beta_0:&(\mathcal{A},\{~,~\})\to U,\\
\alpha_k:&(\mathcal{A},\cdot)\to M_k,\\
\beta_k:&(\mathcal{A},\{~,~\})\to M_k,\\
h_k:&U\to M_k,
\end{align*}
for $k=1,2,\ldots$. These homomorphisms and Lie homomorphisms satisfy $(\ref{eq.alpha})-(\ref{eq.env})$, respectively. In addition, if there exists some linear isomorphism $T_{iso}$ in the morphism of $\mathscr{P}_{env}$, then there also exists an inverse $T_{iso}^{-1}$ of $T_{iso}$. The other morphisms are simply defined by composition maps such that $\mathscr{P}_{env}$ is a category and satisfies the condition 5 in Definition \ref{preq}. 
\end{definition}\bigskip
A character $\chi$ gives the index category $J^\bullet_{env}$ and $F^\bullet_{env}$ as follows. For $J^\bullet_{env}:=\{J_{env}\}$, $J_{env}$ is given as a directed set. For $F^\bullet_{env}:=\{F_{env}\}$, a diagram $F_{env}$ of $J_{env}$ is surjective from $ob(J_{env})\to ob(\mathscr{P}_{env})\backslash \{\mathcal{A}\}$. For the object $U$, $\hbar(\alpha_0\circ \Pi_1)=0$ and $\hbar(\beta _0\circ \Pi_2)=1\slash \sqrt{-1}$ because $\alpha$ is an algebra homomorphism and $\beta$ is a Lie homomorphism, i.e.,
\begin{align}
0=[\alpha(a),\alpha(b)]&=\sqrt{-1}\hbar(\alpha_0\circ \Pi_1)\alpha(\{a,b\}) \label{alpha}\\
\intertext{and} 
\beta_0(\{a,b\})=[\beta_0(a),\beta_0(b)]&=\sqrt{-1}\hbar(\beta_0\circ \Pi_2)\beta_0(\{a,b\}).\label{beta}
\end{align}
Therefore $\chi(U)=|1\slash \sqrt{-1}|=1$. Since $\chi(\mathcal{A})=\chi((\mathcal{A},\cdot))=\chi((\mathcal{A},\{~,~\}))=0$ and $\chi(U)=1$, there at least exist morphisms between them in $F_{env}(J_{env})$. If $\chi(M_i)\ge 1$, $h_i\in \mathscr{P}_{env}(U, M_i)$ exists in $F_{env}(J_{env})$.
\begin{lemma}\label{cenvlem}
$\mathscr{P}_{env}$ is a pre-$\mathscr{Q}$ category.
\end{lemma}
\begin{proof}
It is sufficient to prove that $\mathscr{P}_{env}$ satisfies the condition $(\ref{lie})$ in Definition $\ref{preq}$. For all quantization maps $T_k\in Mor(\mathscr{P}_{env})$ whose domain is $\mathcal{A}$, its codomain is equipped with the Lie bracket $[~,~]_k$ as the commutator.  From $(\ref{alpha})$, $(\ref{beta})$ and so on, we obtain
\begin{align*}
[T_k(f),T_k(g)]_k=\sqrt{-1}\hbar (T_k)T_k(\{f,g\}).
\end{align*} 
\end{proof}
\begin{theorem}
For $(\mathscr{P}_{env},J^\bullet_{env},F^\bullet_{env},\chi)$ given as above, $\mathscr{Q}_{env}:=\mathscr{Q}(\mathscr{P}_{env},J^\bullet_{env},F^\bullet_{env},\chi)$ is a quantization category $\mathscr{Q}_{env}$ with a limit $\mathcal{A}$ of $(J_{env},F_{env})$.
\end{theorem}
\begin{proof}
From the Lemma $\ref{cenvlem}$, $\mathscr{P}_{env}$ is a pre-$\mathscr{Q}$ category. We consider a limit $(M_\infty, \pi)$ of $F_{env}$ as the following commutative diagram. 
\begin{align*}
\xymatrix{
(\mathcal{A},\cdot)\ar[rd]_{\alpha_0} \ar[rdd]_{\alpha_i}&\ar[l]_{~~\Pi_1}\mathcal{A}\ar[r]^{\Pi_2~~}\ar[d]_{\alpha_0 \circ \Pi_1} \ar[d]^{\beta_0 \circ \Pi_2}&\ar[ldd]^{\beta_i} \ar[ld]^{\beta_0}(\mathcal{A}, \cdot, \{~,~\})\\
{}&U\ar@{.>}[d]^{h_i}&{}\\
{}&M_i&{}
}
\end{align*}
Candidates of the limit for a domain of $\mathscr{P}_{env}(\cdot,F_{env}(i))$ for all $i\in ob(J_{env})$ are $\mathcal{A}$ and $(\mathcal{A}, \cdot)$. However, there is no morphism $f\in \mathscr{P}_{env}(\mathcal{A}, (\mathcal{A},\cdot))$ such that $\alpha_0 \circ f=\beta_0\circ \Pi_2=\beta_0$, then $(\mathcal{A}, \cdot)$ is not the limit. Thus, a limit of $F_{env}$ is $(\mathcal{A}, Hom(\mathcal{A},F_{env}(i)))$. Since the quantization map to $\mathcal{A}$ is the only identity, the quantization conditions are satisfied on $\mathcal{A}$ trivially. Therefore $\mathscr{Q}_{env}$ is a quantization category of the Poisson enveloping algebra.
\end{proof}
\section{Categorical Equivalence}\label{secce}
In this section, we discuss the equivalence of the quantization categories appearing in Section \ref{secmr}-\ref{secenv}.\par
\bigskip

Let $\displaystyle \left.\mathscr{Q}_{MR}\vphantom{\big|} \right|_{Mat_{N_i}}$ and $\displaystyle \left.\tilde{\mathscr{Q}}_{DQ}\vphantom{\big|} \right|_{\mathcal{C}^\hbar}$ be those appearing in Corollary $\ref{qmrcor}$ and Corollary $\ref{qdqcor}$. The following proposition is trivially obtained.
\begin{proposition}\label{equivmr}
If there are no linear isomorphisms except for identity maps in $QM_{MR}$ and $QM_{DQ}$, or if $QM^{-1}_{MR}=\{T_\infty^{-1}\}$ and $QM^{-1}_{DQ}=\{(\mathcal{Q}^0)^{-1}\}$, then the quantization categories $\left.\mathscr{Q}_{MR}\vphantom{\big|} \right|_{Mat_{N_i}}$ and $\left.\tilde{\mathscr{Q}}_{DQ}\vphantom{\big|} \right|_{\mathcal{C}^\hbar}$ are equivalent categories. More strictly, they are isomorphic.
\end{proposition}\par
This proposition can be immediately extended to the following proposition.
\begin{proposition}\label{equivmrr}
Let $\sharp N$ be a cardinality of the sequence of $\{N_i\}$ in $\mathscr{Q}_{MR}$ and $\sharp \hbar$ be a cardinality of $\hbar$ in $\tilde{\mathscr{Q}}_{DQ}$. When $\sharp N=\sharp \hbar$, if there are no linear isomorphisms except for idetity maps in $QM_{MR}$ and $QM_{DQ}$, or if $QM^{-1}_{MR}=\{T_\infty^{-1}\}$ and $QM^{-1}_{DQ}=\{(\mathcal{Q}^0)^{-1}\}$, then $\mathscr{Q}_{MR}$ and $\tilde{\mathscr{Q}}_{DQ}$ are equivalent categories. More strictly, they are isomorphic.
\end{proposition}
Recall that $\mathcal{A}(M)$ and $\mathcal{T}_{\mathcal{A}(M)}$ in $\mathscr{Q}_{PQ}$ are isomorphic. So the equivalence with $\mathscr{Q}_{PQ}$ is given as follows.
\begin{proposition}\label{equivmin}
If $QM^{-1}_{MR}=\{T_\infty^{-1}\}$, then $\left.\mathscr{Q}_{MR}\vphantom{\big|} \right|_{Mat_{N_i}}$ and $\mathscr{Q}_{PQ}$ are equivalent categories. More strictly, they are isomorphic.
\end{proposition}
\begin{example}
Let $(M,\omega)$ be a quantizable compact K\"{a}hler manifold, $L$ some very ample line bundle over $M$, and $\mathcal{A}(M)$ a Poisson algebra of smooth functions on $M$. Let $\mathcal{H}^{(m)}=\Gamma_{hol}(M,L^{m})$ be the Hilbert space of holomorphic sections in $L^m$, where $L^m:=L^{\otimes m}$. For $f\in \mathcal{A}(M)$, The Toeplitz operator is defined as
\begin{align*}
T^{(m)}(f):=\Pi^{(m)}\circ \mu_f \circ \Pi^{(m)},
\end{align*}
where $\Pi^{(m)}:L^2(M,L^{m})\to \Gamma_{hol}(M,L^m)$ is a projection and $\mu_f$ is the operator of multiplication by $f$. If $m\to \infty$ then ${\rm End}(\mathcal{H}^{(m)})$ and $\mathcal{A}(M)$ are isomorphic~(See \cite{berezin1}). Therefore, this is the case of the above Proposition $\ref{equivmin}$. 
\end{example}\par
\begin{proposition}\label{equivvvv}
If $QM^{-1}_{DQ}=\{(\mathcal{Q}^0)^{-1}\}$, then $\left.\tilde{\mathscr{Q}}_{DQ}\vphantom{\big|} \right|_{\mathcal{C}^\hbar}$ and $\mathscr{Q}_{PQ}$ are equivalent categories. More strictly, they are isomorphic.
\end{proposition}
\begin{proof}
The proposition is proved immediately in the same reason as Proposition $\ref{equivmin}$. 
\end{proof}
\bigskip
From Proposition $\ref{equivmin}$ and $\ref{equivvvv}$, the following property is derived immediately
\begin{align*}
 \left.\tilde{\mathscr{Q}}_{DQ}\vphantom{\big|} \right|_{\mathcal{C}^\hbar}\simeq \left.\mathscr{Q}_{MR}\vphantom{\big|} \right|_{Mat_{N_i}} \simeq \mathscr{Q}_{PQ}.
\end{align*}
\begin{proposition}\label{notequivpq}
Suppose that $\tilde{\mathscr{Q}}_{DQ}$ and $\mathscr{Q}_{MR}$ have more than three objects each pair of which are not isomorphic, respectively. Then, the quantization category $\mathscr{Q}_{PQ}$ is not categorically equivalent to either $\tilde{\mathscr{Q}}_{DQ}$ or $\mathscr{Q}_{MR}$.
\end{proposition}
\begin{proof}
From Proposition \ref{apeB1} in Appendix \ref{apb}, $\mathscr{Q}_{PQ}$ and $\mathscr{B}$ in Proposition \ref{apeB1} are equivalent categories. Since $\tilde{\mathscr{Q}}_{DQ}$ has more than three objects each pair of which are not isomorphic, at least $\tilde{\mathscr{Q}}_{DQ}$ and $\mathscr{B}$ in Proposition \ref{apeB1} are not equivalent categories from Proposition \ref{apeB2} in Appendix \ref{apb}. Thus $\mathscr{Q}_{PQ}$ and $\tilde{\mathscr{Q}}_{DQ}$ are not equivalent categories. The case of $\mathscr{Q}_{MR}$ is proved in the same way.
\end{proof}
\begin{proposition}\label{envdq}
$\mathscr{Q}_{env}$ and $\tilde{\mathscr{Q}}_{DQ}$ are not equivalent categories.
\end{proposition}
\begin{proof}
We consider the following diagram.
\begin{align*}
\xymatrix{
(\mathcal{A},\cdot)\ar[rd]_{\alpha_0} \ar[rdd]_{\alpha_i}&\ar[l]_{~~\Pi_1}\mathcal{A}\ar[r]^{\Pi_2~~}\ar[d]_{\alpha_0 \circ \Pi_1} \ar[d]^{\beta_0 \circ \Pi_2}&\ar[ldd]^{\beta_i} \ar[ld]^{\beta_0}(\mathcal{A}, \{~,~\})\\
{}&U\ar[d]^{h_i}&{}\\
{}&M_i&{}
}
\end{align*}
Suppose that $\mathscr{Q}_{env}$ and $\tilde{\mathscr{Q}}_{DQ}$ are equivalent categories: that is, there exist functors $F:\mathscr{Q}_{env}\to \tilde{\mathscr{Q}}_{DQ}$ and $G:\tilde{\mathscr{Q}}_{DQ}\to \mathscr{Q}_{env}$, $FG\simeq id_{\tilde{\mathscr{Q}}_{DQ}}$ and $GF\simeq id_{\mathscr{Q}_{env}}$. Then $\alpha_0\circ \Pi_1$ and $\beta _0\circ \Pi_2$ must be mapped to $\tilde{\mathscr{Q}}_{DQ}(F(\mathcal{A}), F(U))$ by $F$:
\begin{align*}
F(\alpha_0 \circ \Pi_1),~F(\beta_0 \circ \Pi_2)\in \tilde{\mathscr{Q}}_{DQ}(F(\mathcal{A}), F(U)).
\end{align*}
Recall that $QM_{DQ}^{-1}=\emptyset$ or $QM^{-1}_{DQ}=\{(\mathcal{Q}^0)^{-1}\}$, that is, if $\tilde{\mathscr{Q}}_{DQ}(X,Y)\neq \emptyset$ there is the unique morphism from $X$ to $Y$ for arbitrary $X,Y\in ob(\tilde{\mathscr{Q}}_{DQ})$. So $F(\alpha_0\circ \Pi_1)=F(\beta \circ \Pi_2)$, and
\begin{align*}
GF(\alpha_0 \circ \Pi_1))=GF(\beta_0 \circ \Pi_2))\in \mathscr{Q}_{env}(GF(\mathcal{A}), GF(U)),
\end{align*}
for arbitrary functor $G:\tilde{\mathscr{Q}}_{DQ}\to \mathscr{Q}_{env}$. A diagram of a natural transformation $\mu:GF\to id_{\mathscr{Q}_{env}}$ is given by 
\begin{align*}
\xymatrix{
GF(\mathcal{A}) \ar[r]_{\mu_{\mathcal{A}}} \ar[d]_{GF(\alpha_0\circ \Pi_1)=GF(\beta_0\circ \Pi_2)}& \mathcal{A} \ar@<0.5ex>[d]_{\alpha_0 \circ \Pi_1~} \ar@<-0.5ex>[d]^{~\beta_0\circ \Pi_2}\\
GF(U)\ar[r]_{\mu_U}&U.
}
\end{align*}
However, since $\alpha_0 \circ \Pi_1\neq \beta_0 \circ \Pi_2$ in $\mathscr{Q}_{env}$, there are no natural isomorphisms that satisfies the commutatibility of the diagram. This is a contradiction. Thus, $\mathscr{Q}_{env}$ is not categorically equivalent to $\tilde{\mathscr{Q}}_{DQ}$.
\end{proof}
\begin{proposition}
$\mathscr{Q}_{env}$ is categorically equivalent to neither $\mathscr{Q}_{PQ}$ nor $\mathscr{Q}_{MR}$.
\end{proposition}
\begin{proof}
Since $\mathscr{Q}_{PQ}$ and $\mathscr{Q}_{MR}$ has the unique morphism between its arbitrary objects, the cases of them are proved in the same manner as Proposition $\ref{envdq}$.
\end{proof}
\section{Conclusions and Discussions}\label{seccon}
In this article, we discuss a category of quantization of Poisson manifolds or Poisson algebras as a subcategory of $R${\rm Mod}, but its objects are commutative and noncommutative algebras. We define the quantization category as a generalization of quantizations of the Poisson algebra, and show that this category contains categories of some known quantizations of the Poisson algebra. We also discuss relationships between the categories of various types of quantizations. \par
The pre-$\mathscr{Q}$ category $\mathscr{P}$ is defined by choosing a fixed Poisson algebra $\mathcal{A}$ and algebras $M_i$ of $\mathcal{A}$'s quantized algebras as objects, and choosing quantization linear maps $T_i:\mathcal{A}\to M_i$ and linear maps between each $M_i$ as morphisms. If a morphism is a linear isomorphism, then its inverse is also a morphism. Each quantization map $T_i$ has a noncommutative parameter $\hbar(T_i)$. The character $\chi(M_i)$ is introduced as the maximum absolute value of $\hbar(T_i)$. The index category $J^\bullet$ and a set of functors $F^\bullet$ are determined by the noncommutative character $\chi$ to consider the classical limit. In addition to these structures, we defined the quantization category as being equipped with homomorphisms of the algebra between Poisson brackets in $\mathcal{A}$ and Lie brackets in the limit determined by $F^\bullet$ of $J^\bullet$. As concrete examples, the quantization category of matrix regularization (including Berezin-Toeplitz quantization), strict deformation quantization, prequantization and Poisson enveloping algebra are constructed. The equivalence or non-equivalence of these categories is also discussed. In particular, we show that when $\sharp N=\sharp \hbar$ the quantization category of matrix regularization and the quantization category of strict deformation quantization are equivalent categories if there are no isomorphisms except for identities in $QM_{MR}$ and $QM_{DQ}$, or if $QM^{-1}_{MR}=\{T_\infty^{-1}\}$ and $QM_{DQ}^{-1}=\{(\mathcal{Q}^0)^{-1}\}$. In addition, we show the equivalence between $\left.\tilde{\mathscr{Q}}_{DQ}\vphantom{\big|} \right|_{\mathcal{C}^\hbar}$, $\left.\mathscr{Q}_{MR}\vphantom{\big|} \right|_{Mat_{N_i}}$ and $\mathscr{Q}_{PQ}$ if there are no isomorphisms except for identities in $QM_{MR}$ and $QM_{DQ}$, or if $QM^{-1}_{MR}=\{T_\infty^{-1}\}$ and $QM_{DQ}^{-1}=\{(\mathcal{Q}^0)^{-1}\}$. For example, this condition is satisfied for compact K\"{a}hler manifolds in the case of Berezin-Toeplitz quantization. On the other hand, it is shown that the quantization category of Poisson enveloping algebra is categorically equivalent to neither $\mathscr{Q}_{MR}$, nor $\mathscr{Q}_{PQ}$ nor $\tilde{\mathscr{Q}}_{DQ}$. \bigskip

We have focused on the quantization category which contains one quantization procedure. However, we define the quantization category such that it is possible to include multiple types of quantization theories. For example, if the intersection of index categories $J^\bullet_1$ and $J^\bullet_2$ of two quantization categories $\mathscr{Q}_1$ and $\mathscr{Q}_2$ with the same Poisson algebra $\mathcal{A}$ is empty, a category consisting of the sum of $\mathscr{Q}_1$ and $\mathscr{Q}_2$ is also a quantization category. Here, the category made up by this summation is a category whose object set is the union of the object sets of $\mathscr{Q}_1$ and $\mathscr{Q}_2$, and its set of morphisms is created from the union of morphisms of $\mathscr{Q}_1$ and $\mathscr{Q}_2$ and adding composite maps to the union so that the whole becomes a category. For example, a category created by the sum of $\mathscr{Q}_{MR}$ and $\tilde{\mathscr{Q}}_{DQ}$ is a quantization category. One of the future tasks will be to examine the concrete construction and to study properties of such a category made up of the sum of quantization categories. It is also necessary to consider the sums of more complex categories whose $J_1^\bullet$ and $J_2^\bullet$ are not disjoint. Such researches should be done as a next step.\par
\bigskip

In this paper, we formulated a quantization category by adopting $(1),(2)$ and $(3)$ among the conditions by Dirac enumerated at the beginning of this article. However, we can choose other combinations. Therefore, the quantization category studied in this paper might be an example of a series of quantization categories that have a variety of quantization conditions. The task of investigating such a large area of quantization categories remains for future work. 
\bigskip

Finally, we will consider potential applications of the quantization category to physics.\\
The universe we live is classically described as a vector bundle.
The base manifold is a Riemannian manifold. 
The fibers are for the electromagnetic field, non-Abelian gauge fields, matter fields and so on.
In the case of the particle physics given by a Hamiltonian formulation of mechanics, 
a cotangent bundle over the Riemannian manifold is its geometry. 
The cotangent bundle is the Poisson manifold.
When we consider $M$ a cotangent bundle over a Riemannian manifold,
then the $A(M)$ in $\mathscr{Q}$ corresponds with a classical physics.
In that case, the character $\chi $ or $\hbar$ should be chosen as the Planck constant or energy scale.
To make concrete predictions or to clarify physical properties, we have to import further structures 
into the quantization category.
The findings of several previous studies might be useful when introducing physical structures.
For example, Ojima and Takeori \cite{Ojima:2006rf} 
studied the correspondence between 
classical and quantum physics, which is called Micro-macro duality, by categorical approach.
Alternatively, 
categorical approaches to quantum mechanics have shown how to describe the 
fine structure of physics as functors.\cite{QM1,QM2}
\\

As an example, let us consider the IKKT
matrix model or noncommutative gauge theory in the context of $\mathscr{Q}_{MR}$ \cite{IKKT1,IKKT2,IKKT3}.
In Section \ref{secmr}, we considered matrix regularization.
The classical IKKT matrix model is regarded as a matrix regularization of the
type IIB string theory with the Schild gauge.
The bosonic part of the Lagrangian is given as
\begin{align*}
\{ X^{\mu} , X^{\nu} \}\{ X^{\mu} , X^{\nu} \},
\end{align*}
where $X^{\mu}$ is a map from a parameter space to a world sheet and
$\{ X^{\mu} , X^{\nu} \}$ is the Poisson bracket on the world sheet.
Using a determinant of the world sheet metric $g$ and Levi-Civita symbol 
$\epsilon^{ab}$, the Poisson bracket is defined as
\begin{align*}
\{ X^{\mu} , X^{\nu} \}:= 
\frac{1}{\sqrt{g}}\epsilon^{ab} \partial_a X^{\mu} \partial_b X^{\nu} .
\end{align*}
The quantization map $T_i$ of the quantization category of matrix regularization in Section $\ref{secmr}$ 
maps this Poisson bracket to a commutator
\begin{align*}
[ X^{\mu} , X^{\nu}] := X^{\mu} X^{\nu} - X^{\nu} X^{\mu},
\end{align*}
up to higher order terms of $\hbar$. Then the Lagrangian is also obtained as
\begin{align*}
[ X^{\mu} , X^{\nu} ][ X^{\mu} , X^{\nu} ],
\end{align*}
at the limit $Mat_{\infty}$. This matrix model is also regarded as a noncommutative 
$U(1)$ gauge theory on noncommutative 
Euclidean space at the limit of the quantization category of matrix regularization.
In this context, type IIB string theory with Schild gauge is
defined on the object $A(M)$ and the IKKT matrix model 
is defined on the object $Mat_\infty$.
\\
\bigskip

In this article, only the ordinary Poisson structure 
has been considered for quantization.
However, there are many other types of classical mechanics, such as
Nambu mechanics \cite{nambu}. 
To attack the quantization problem of the membrane theory, quantization of the Nambu bracket has been shown to be an effective approach.
For this purpose, the Nambu bracket should be replaced with the Lie 3 bracket by the 
quantization morphism \cite{3bra1,3bra2,3bra3,3bra4,3bra5,3bra6}.
The category of quantization defined in this article is naturally generalized to such a quantization type. This would be a suitable problem to address in the next stage of an investigation.\par
\bigskip
The quantization category proposed in this article involves many basic or applied problems, including pure mathematical and physical problems, as mentioned above. All of these should be solved in the future.

%
%
\section*{Acknowledgements}
\noindent 
A.S.\ was supported by JSPS KAKENHI Grant Number 16K05138.
The authors wish to thank A. Cardona and A. Yoshioka for their useful comments and discussions. We thank the anonymous reviewers for their careful reading of our manuscript and their many insightful comments and pointing out mistakes.

\appendix
\section{Definition of $\tilde{O}$}\label{ap1}
Since we have not defined a norm for objects of pre-$\mathscr{Q}$ category $\mathscr{P}$ in this paper, Landau symbol $O$ does not make sense. So we define an order $\tilde{O}$ by $x\in \mathbb{R}$ with the Euclidean norm.

\begin{definition}
Let $\mathcal{M}$ be a $R$-module for a commutative algebra $R$ over $\mathbb{C}$. Let $f_i$ be a complex valued function such that
\begin{align*}
\lim_{x\to 0} \frac{f_i(xz)}{x}=0,
\end{align*}
where $x\in \mathbb{R}$ and $z\in \mathbb{C}$. For $a_i\in \mathcal{M}$ which is independent of $z\in \mathbb{C}$, we denote the element described as $\sum_i f_i(z)a_i\in \mathcal{M}$ by $\tilde{O}(z)$.
\end{definition}
From this definition, the term of $\tilde{O}(\hbar^{1+\epsilon})$ in $(\ref{lie})$ is also $0$ when $\hbar=0$. Note that   $\hbar$ is not necessarily continuous in our definition of the pre-$\mathscr{Q}$ category.
\section{Structures of Simple Categories}\label{apb}
In this section, we discuss structures of simple categories used in Section \ref{secce} and Section \ref{secdq}.
\begin{proposition}\label{apeB1}
For categories $\mathscr{A}$ and $\mathscr{B}$ given as
\begin{align*}
\xymatrix{
A\ar[rd]^{g~~~~~~~~~~~~\mbox{and}} \ar@<0.5ex>[d]^{f}&{}&{~~~~~}&A^\prime \ar[d]^{g^\prime}\\
B\ar@<0.5ex>[u]^{f^{-1}}\ar[r]_h&C&{~~~~~}&C^\prime,
}
\end{align*}
$\mathscr{A}$ and $\mathscr{B}$ are equivalent categories. Here 
\begin{align*}
ob(\mathscr{A})&=\{A,B,C\},\\
Mor(\mathscr{A})&=\{f,f^{-1},g,h, id_A, id_B, id_C\},\\
ob(\mathscr{B})&=\{A^\prime, C^\prime\},\\
Mor(\mathscr{B})&=\{g^\prime , id_{A^\prime}, id_{C^\prime}\}.
\end{align*}
\end{proposition}
\begin{proof}
We define functors $F:\mathscr{A}\to \mathscr{B}$ and $G:\mathscr{B}\to \mathscr{A}$ by
\begin{align*}
F(A)&=F(B)=A^\prime ,\\
F(C)&=C^\prime ,\\
F(f)&=F(f^{-1})=id_{A^\prime},\\
F(g)&=F(h)=g^\prime ,\\
\\
G(A^\prime)&=A,\\
G(C^\prime)&=C,\\
G(g^\prime)&=g.
\end{align*}
Since $A$ and $B$ are isomorphic by $f$ and $f^{-1}$, trivially, $FG\simeq id_{\mathscr{B}}$ and $GF\simeq id_{\mathscr{A}}$. 
\end{proof}
\begin{proposition}\label{apeB2}
For categories $\mathscr{A}$ and $\mathscr{B}$ given as
\begin{align*}
\xymatrix{
A\ar[rd]^{g~~~~~~~~~~~~\mbox{and}} \ar[d]_{f}&{}&{~~~~~}&A^\prime \ar[d]^{g^\prime}\\
B\ar[r]_h&C&{~~~~~}&C^\prime,
}
\end{align*}
$\mathscr{A}$ and $\mathscr{B}$ are not equivalent categories. Here 
\begin{align*}
ob(\mathscr{A})&=\{A,B,C\},\\
Mor(\mathscr{A})&=\{f,g,h, id_A, id_B, id_C\},\\
ob(\mathscr{B})&=\{A^\prime, C^\prime\},\\
Mor(\mathscr{B})&=\{g^\prime , id_{A^\prime}, id_{C^\prime}\}.
\end{align*}
\end{proposition}
\begin{proof}
If $\mathscr{A}\simeq \mathscr{B}$, then there exists a natural isomorphism $\theta:GF\to id_{\mathscr{A}}$ for $F:\mathscr{A}\to \mathscr{B}$ and $G:\mathscr{B}\to \mathscr{A}$. That is, for all objects $X,Y\in \mathscr{A}$, if there exists $m\in \mathscr{A}(X,Y)$ the following diagram commutes.
\begin{align*}
\xymatrix{
GF(X)\ar@<0.5ex>[r]^{\theta_X}\ar[d]_{GF(m)}&X\ar@<0.5ex>[l]^{\theta_X^{-1}}\ar[d]^{m}\\
GF(Y)\ar@<0.5ex>[r]^{\theta_Y}&Y\ar@<0.5ex>[l]^{\theta_Y^{-1}}.
}
\end{align*} 
However, there are no natural isomorphisms such that the following diagrams commute  
\begin{align*}
\xymatrix{
GF(A)\ar[r]^{\theta_{A}} \ar[d]_{GF(f)}&A \ar[d]^{f} & {}&GF(A)\ar[r]^{\theta_{A}} \ar[d]_{GF(g)}&A \ar[d]^{g} \\
GF(B)\ar[r]_{\theta_{B}}&B,&{}&GF(C)\ar[r]_{\theta_{C}}&C.
}
\end{align*}
Indeed, at least there exists an object $X\in ob(\mathscr{A})$ such that $GF(X)=GF(Y)$ for $X\neq Y$, and $Mor(\mathscr{A})$ has no isomorphisms.
\end{proof}
For the category $\mathscr{A}$ in Proposition \ref{apeB2}, we consider the case that there exists $g^{-1}\in \mathscr{A}(C,A)$ such that $g\circ g^{-1}=id_{C}, g^{-1}\circ g=id_{A}$.
\begin{proposition}\label{apeB3}
Let $\mathscr{A}$ be a category in Proposition \ref{apeB2}. If $g\in \mathscr{A}(A,C)$ is an isomorphism, then $\mathscr{A}^\prime$ given by 
\begin{align*}
ob(\mathscr{A}^\prime)&:=\{A,B,C\}\\
Mor(\mathscr{A}^\prime)&:=\{f,h,g,g^{-1},f^\prime,h^\prime,e_B, id_A, id_B, id_C\}
\end{align*}
is a category. Here
\begin{align*}
f^\prime&=g^{-1}\circ h\\
h^\prime&=f\circ g^{-1}\\
e_B&=f\circ g^{-1}\circ h.
\end{align*}
\end{proposition}
\begin{proof}
We show that its composition is closed for all morphisms. The composition operation table is given as
\begin{align*}
\begin{array}{|c|c|c|c|c|c|c|c|c|}
    \hline
              &     f&    h&    g&    g^{-1}&    f^\prime&    h^\prime&    e_B\\
    \hline
       \circ f&    -&    g&    -&    -&    id&    -&    f\\
    \hline
       \circ h&    -&   - &    -&    f^\prime&   - &   e_B &   - \\
    \hline
       \circ g&    -&    -&    -&    id&   - &   f &   -     \\
    \hline
       \circ g^{-1}&   h^\prime &   - &   id & -   & -   & -       &-    \\
    \hline
       \circ f^\prime&   e_B &  -  &  h  & -   & -   & -       & -    \\
    \hline
       \circ h^\prime&  -  &  id  & -   & -   & g^{-1}   &    -    &h^\prime  \\
    \hline
       \circ e_B&  -  &  h  &  -  &  -  & f^\prime   & -   & e_B       \\
    \hline
\end{array}
\end{align*}
Thus, $\mathscr{A}^\prime$ is a category.
\end{proof}

\end{document}